\documentclass[a4paper,12pt]{scrartcl} 

\usepackage[english]{babel}
\usepackage[utf8]{inputenc}
\usepackage[T1]{fontenc}
\usepackage{eucal}
\usepackage{microtype}
\usepackage{csquotes}
\usepackage{amsthm}
\usepackage{mathtools}
\usepackage{amssymb,amsfonts,stmaryrd}
\usepackage{tikz}
\usetikzlibrary{cd}
\usepackage{xparse}
\usepackage[backend=biber,giveninits=true,maxnames=5,doi=true,isbn=false,url=false,sortcites,style=phys,biblabel=brackets,sorting=nyt]{biblatex} 
\usepackage{authblk}

\usepackage[all]{xy}

\RequirePackage{hyperref}
\hypersetup{breaklinks,unicode=true}

\theoremstyle{plain}
\newtheorem{theorem}{Theorem}
\newtheorem*{theorem*}{Theorem}
\newtheorem{lemma}[theorem]{Lemma}
\newtheorem*{lemma*}{Lemma}
\newtheorem{proposition}[theorem]{Proposition}
\newtheorem*{proposition*}{Proposition}
\newtheorem{corollary}[theorem]{Corollary}

\newtheorem*{corollary*}{Corollary}
\theoremstyle{definition}
\newtheorem{definition}{Definition}
\newtheorem*{definition*}{Definition}

\newtheorem*{example*}{Example}
\theoremstyle{remark}
\newtheorem{remark}{Remark}
\newtheorem*{remark*}{Remark}

\newtheorem*{conjecture*}{Conjecture}
\newtheorem{problem}{Problem}
\newtheorem*{problem*}{Problem}

\usepackage{cleveref}

\newcommand*{\RR}{\mathbb{R}}

\newcommand*{\dd}{\mathrm{d}}
\DeclareMathOperator{\Id}{Id}
\DeclareMathOperator{\im}{im}

\newcommand*{\contr}[1]{\iota_{#1}}
\newcommand*{\liedv}[1]{\mathcal{L}_{#1}}
\newcommand*{\Reeb}{\mathcal{R}}

\title{Inverse problem and equivalent contact systems}

\author[1,2]{Manuel de León\thanks{\textit{email}:
mdeleon@icmat.es (ORCID: 0000-0001-8201-1624)}}
\author[3]{Jordi Gaset\thanks{\textit{email}:
jordi.gaset@unir.net (ORCID: 0000-0001-8796-3149)}} 
\author[1]{Manuel Lainz\thanks{\textit{email}:
manuel.lainz@icmat.es. (ORCID: 0000-0002-2368-5853)}}

\affil[1]{Instituto de Ciencias Matemáticas (CSIC-UAM-UC3M-UCM) \protect\\
Calle Nicolás Cabrera, 13-15, Campus Cantoblanco, UAM, 28049 Madrid, Spain}
\affil[2]{Real Academia de Ciencias Exactas, Físicas y Naturales\protect\\
Calle Valverde, 22, 28004, Madrid, Spain }
\affil[3]{Escuela Superior de Ingeniería y Tecnología, Universidad Internacional de La Rioja, Spain.\protect}

\date{\today}

\begin{document}

\maketitle
\begin{abstract}
We present several results on the inverse problem  and equivalent contact Lagrangian systems. These problems naturally lead to consider smooth transformations on the $z$ variable (i.e., reparametrizations of the action). We present the extended contact Lagrangian systems to formalize this notion. With this structure we define horizontal equivalence of Lagrangians, which generalizes the symplectic case. We also present some results on the inverse problem for extended contact systems.
\end{abstract}
\setcounter{tocdepth}{2}

{
\def\baselinestretch{0.97}
\small
\def\addvspace#1{\vskip 1pt}
\parskip 0pt plus 0.1mm
\tableofcontents
}

\section{Introduction}
As it is well known, given a Lagrangian function $L:TQ \to \RR$ one obtains the \emph{Euler-Lagrange} vector field $\xi_L$, which is a second order differential equation such that its integral curves $c(t)$ are solutions of the Euler-Lagrange equation for $L$.
However, we might ask ourselves the inverse question. Consider a SODE $\xi$ on $TQ$ such that its integral curves $c(t)$ satisfy the equation
\begin{equation}
    \ddot{c}^i(t) = f^i(c(t), \dot{c}(t))
\end{equation}
for some local functions $f^i:U \subseteq TQ \to \mathbb{R}$. The so-called \emph{inverse problem of calculus of variations} ask if the equation above can be derived through a Lagrangian $L$. Namely, is $\xi = \xi_L$ for some Lagrangian $L$? A partial solution to this problem is given by Helmholtz conditions~\cite{Crampin1981a}. The inverse problem can be solved if and only if there exist functions $g_{ij}$ such that
\begin{subequations}
    \begin{align}
        \det{g_{ij}} &= 0\\
        g_{ij} &= g_{ji} \\
        \frac{\partial g_{ij}}{\partial q^k} &= \frac{\partial g_{ik}}{\partial q^j}\\
        \frac{\dd g_{ij}}{\dd t} + \frac{1}{2}\frac{\partial f^k}{\partial \dot{q}^j} g_{ik} + \frac{1}{2}\frac{\partial f^k}{\partial \dot{q}^i} g_{kj} &= 0  \\
        g_{ik} \left(\frac{\dd }{\dd t} \left(\frac{\partial f^k}{\partial \dot{q}^j}\right) -2 \frac{\partial f^k}{\partial q^j}  - \frac{1}{2} \frac{\partial f^l}{\partial \dot{q}^j} \frac{\partial f^k}{\partial \dot{q}^l}\right) &= 
        g_{jk} \left(\frac{\dd }{\dd t} \left(\frac{\partial f^k}{\partial \dot{q}^i}\right) -2 \frac{\partial f^k}{\partial q^i}  - \frac{1}{2} \frac{\partial f^l}{\partial \dot{q}^i} \frac{\partial f^k}{\partial \dot{q}^l}\right).
    \end{align}
\end{subequations}

These conditions can be characterized geometrically as follows~\cite{Crampin1981a}: the inverse problem can be solved for a SODE $\xi$ if and only if there exist a $2$-form $\omega$ on $TQ$ of maximal rank such that $\liedv{\xi} \omega = 0$ and such that all vertical subspaces are Lagrangian both for $\omega$ and for $\contr{H} \dd \omega$ for any horizontal vector field $H$.

A related problem is the problem of \emph{equivalent Lagrangians} \cite{Ranada1991} which can be stated as follows. Given a Lagrangian $L$, find every Lagrangian $\tilde{L}$ that have the same Euler-Lagrange vector field, that is, $\xi_L = \xi_{\tilde{L}}$. A sufficient condition is that their difference is a \emph{total derivative}. 


Nonetheless, there are systems that cannot be described through Euler-Lagrange equations. Some interesting examples \cite{Goto2016,RMS-2017,Ciaglia2018, Bravetti2018, Simoes2020b,GM-2021}, such as thermodynamic systems at equilibrium or several mechanical systems with friction can be modeled through a contact Hamiltonian system \cite{Geiges2008,BHD-2016,Bravetti2017,deLeon2019a}. This kind of systems can be also modeled with a Lagrangian $L:TQ \times \mathbb{R} \to \mathbb{R}$ that, apart from the positions and velocities, it depends on an extra variable $z$ that can be interpreted as the action \cite{deLeon2019,GGMRR-2019b}. From this Lagrangian we obtain a Herglotz vector field $\xi_L$ which is a SODE on $TQ\times \mathbb{R}$ (meaning that its integral curves $(c,v,\zeta)(t)$ satisfy $\dot{c}=v)$. Moreover, the integral curves $(c,\dot{c},\zeta)(t)$ of $\xi_L$ are precisely the ones that satisfy the Herglotz equations
\begin{subequations}\label{eq:herglotz}
\begin{align}
    \frac{\partial L}{\partial q^i}(c(t), \dot{c}(t),\zeta(t)) - \frac{\dd }{\dd t} \left(\frac{\partial L}{\partial \dot{q}^i} (c(t), \dot{c}(t)),\zeta(t) \right) \label{eq:herglotz1} &= 
    \frac{\partial L}{\partial \dot{q}^i} \frac{\partial L}{\partial z} \\
    \dot{z} &= L. \label{eq:herglotz2}
\end{align}
\end{subequations}

The Lagrangian formulation can also be derived through a variational principle, the so-called \emph{Herglotz principle} in which the action is defined through a non-autonomous ODE \cite{Herglotz1930}. In Section~\ref{sec:prelim} we explain the theory of contact geometry and contact Hamiltonian and Lagrangian systems necessary for the development of this paper.

Given the recent interest on contact Lagrangian systems, we think it will be useful to have a tool, similar to the Helmholtz conditions. That is, a straightforward procedure to decide weather a given SODE on $TQ \times \mathbb{R}$ are Herglotz equations or not. In Section~\ref{sec:inv_problem1}, we naively try to formulate the inverse problem. However, this problem turns out to be trivial. Indeed, by looking at~\eqref{eq:herglotz2}, one sees that if a SODE on $TQ\times \mathbb{R}$ is a Herglotz vector field, then the Lagrangian is necessarily the coefficient of $\partial / \partial z$. Nonetheless, we obtain an interesting geometrical characterization of Herglotz vector fields that can be compared to the one obtained by~\cite{Sarlet1982}.

The equivalences between Lagrangian systems suffers from a similar problem. Our aim is to answer the question: given a regular Lagrangian $L:TQ \times \mathbb{R} \to \mathbb{R}$, find all regular Lagrangians that share the same dynamics. But what do we mean by \enquote{the same}? Taken literally, we would need to find all the Lagrangians $\bar{L}:TQ\times \RR \to \RR$ for which $\xi_L = \xi_{\bar{L}}$. This problem, however, would be trivial. Since $\xi_L(z) = L$ (hence $\xi_{\bar{L}}(z)=\bar{L}$), $L$ and $\bar{L}$ are equivalent if and only if they are equal. 

In order to obtain a more useful statement of inverse problem and equivalent Lagrangians in the contact setting, we need a weaker version of the problems. For this, we analyze what means for two SODEs on $TQ \times \mathbb{R}$ to represent the same dynamics. Thinking on the physical examples we are familiar with, the $q^i$ and $\dot{q}^i$ variables usually represent quantities that are directly observable on the system (positions and velocities). However, the variable $z$ represents the action or a thermodynamic potential such as the entropy or the internal energy, none of which are directly measurable. Hence, we think that a notion of equivalence \enquote{up to change of variables in $z$} is interesting for these applications. In Sections~\ref{sec:extended_systems} we formalize this notion with the introduction of \emph{extended contact systems}.

A summary of results of equivalent contact Hamiltonian is presented in section \ref{sec:equiv_ham}. Since contact Lagrangian systems (even the extended ones) are contact Hamiltonian systems, this will be general results with a clear geometrical meaning.

In Section~\ref{equiv.extended}, we study equivalences between contact Lagrangian systems, keeping in mind that those systems have two geometrical structures: a contact structure and an extended tangent bundle structure, which have been studied on the previous section. Analyzing the interaction of the equivalences on both structures allows us to obtain a notion of equivalence that generalizes the one for the Euler-Lagrange equations in the case that the Lagrangian does not depend on $z$. Some examples are presented exploring possible applications.

Finally, in Section~\ref{sec:inv_problem2} we study the inverse problem \enquote{up to change of variables in $z$}. A Helmholtz-like conditions and a geometric characterization are obtained.

\section{Contact Lagrangian systems}\label{sec:prelim}

(See \cite{deLeon2019a, deLeon2017, GGMRR-2019b} for more details.)

Let $P$ be a $(2n+1)$-dimensional manifold. A \emph{contact structure} in $P$ is a non-degenerate $1$-form $\eta\in\Omega^1(P)$ such that $\eta\wedge(\dd\eta)^n$ is a volume form. The \emph{Reeb vector field} of $\eta$ is the unique vector field $\mathcal{R}\in\mathfrak{X}(P)$ such that
$$
\iota_\mathcal{R}\eta=1\,;\quad \iota_\mathcal{R}\dd\eta=0\,.
$$
Around any point of a contact manifold, there exist \emph{Darboux coordinates}, $(q^i, p_i, z)$, such that
\begin{equation}
    \eta = \dd z  - p_i \dd q^i
\end{equation}
and
\begin{equation}
    \Reeb = \frac{\partial }{\partial z}.
\end{equation}

A \emph{contact Hamiltonian system} on $P$ is the pair $(\eta,H)$, where $H\in C^\infty(P)$. The Hamiltonian vector field $X_H\in\mathfrak{X}(P)$ is the unique solution of the equations
\begin{subequations}\label{eq:Herglotz_Hamiltonian}
\begin{align}
    \iota_{X_H}\dd \eta&=\dd H-\left(\mathcal{L}_{\mathcal{R}}H\right) \eta\,,
    \\
    \iota_{X_H}\eta&=-H\,,
\end{align}    
\end{subequations}
or, equivalently,
\begin{align*}
    \mathcal{L}_{X_H}\eta&=-\left(\mathcal{L}_{\mathcal{R}}H\right) \eta\,,
    \\
    \iota_{X_H}\eta&=-H\,.
\end{align*}
In Darboux coordinates
\begin{equation}
    X_H = \frac{\partial H}{\partial q^i} \frac{\partial }{\partial q^i} - \left( \frac{\partial H}{\partial p_i} + p_i \frac{\partial H}{\partial z} \right) \frac{\partial }{\partial p_i} + \left(p_i \frac{\partial H}{\partial p_i} - H \right) \frac{\partial }{\partial z}
\end{equation}

A \emph{contact Lagrangian system} is given by a Lagrangian function $L:P=TQ \times \mathbb{R} \to \mathbb{R}$, where $Q$ is an $n$-dimensional manifold. The canonical endomorphism $S$ and the Liouville vector field $\Delta$ of $TQ$ extend to $TQ \times \mathbb{R}$ naturally due to the trivial factorization into $TQ$ and $\mathbb{R}$. If we consider coordinates $(q^i,\dot{q}^i,z)$ on $TQ \times \mathbb{R}$ then, their local expression are
$$
S=\dd q^i\otimes \frac{\partial}{\partial\dot{q}^i}\,;\quad \Delta=\dot{q}^i\frac{\partial}{\partial\dot{q}^i}\,.
$$
Given a contact Lagrangian system we can construct a contact Lagrangian form and a Lagrangian energy 
$$
\eta_L=\dd z-S^*\dd L\,;\quad E_L=\Delta(L)-L\,.
$$
$\eta_L$ is a contact form if, and only if, $L$ is regular, that is, the matrix $\left(\frac{\partial^2L}{\partial \dot{q}^i\dot{q}^j} \right)$ is non-singular everywhere. Then, $(\eta_L,E_L)$ form a contact Hamiltonian system, whose Hamiltonian vector field $\xi_L$ is called the \emph{Herglotz vector field} of $L$.

There is a variational principle behind a contact Lagrangian systems, called \emph{Herglotz's variational problem}. Given the space $\Omega(q_0,q_1)$ of paths $\gamma:[0,1] \to Q$ such that $\gamma(0)=q_0$ and $\gamma(1)=q_1$, we define
\begin{equation}
    Z_{L,z_0}: \Omega(q_0,q_1) \to \mathcal{C}^\infty([0,1] \to \mathbb{R})
\end{equation}
such that for each curve $\gamma \in \Omega(q_0, q_1)$, $Z_{L,z_0}(\gamma)$ is the curve that solves the initial value problem
\begin{equation}
    \begin{cases}
        \frac{\dd Z_{L,z_0}(\gamma)}{\dd t} &= L(\gamma,\dot{\gamma}, Z_{L,z_0}(\gamma)),\\
       Z_{L,z_0}(\gamma)(0) &= z_0.
    \end{cases}
\end{equation}

We now define the action
\begin{equation}
    \begin{split}
        \mathcal{A}: \Omega(q_0,q_1) &\to \mathbb{R},\\
        \gamma &\to Z_{L,z_0}(\gamma)(1) - z_0 = \int_0^1 L(\gamma, \dot{\gamma}, Z_{L,z_0}(\gamma)) \dd t.
    \end{split}
\end{equation}
The curves that minimize $\mathcal{A}$ are the ones that satisfies \emph{Herglotz's equations}
    \begin{subequations}
        \begin{align}
            \frac{\partial L}{\partial q^i} - \frac{\dd }{\dd t} \frac{\partial L}{\partial \dot{q}^i}  &= 
           \frac{\partial L}{\partial \dot{q}^i} \frac{\partial L}{\partial \zeta} \\
            \dot{z} &= L.    
        \end{align}
    \end{subequations}
These are the same equations as \eqref{eq:Herglotz_Hamiltonian} in coordinates for the system $(\eta_L,E_L)$.

\section{Inverse problem for contact Lagrangian systems}\label{sec:inv_problem1}

Naively, we can formulate the inverse problem for contact Lagrangian systems as follows:
\begin{problem}[Inverse problem]
    Given a SODE $\xi$ on $TQ \times \RR$, determine if there exists a contact Lagrangian $L:TQ \times \RR \to \RR$ such that $\xi = \xi_L$.
\end{problem}
Nevertheless, this problem is much more restrictive than its symplectic equivalent. Indeed, the local form of the vector field $\xi_L$ is
\begin{equation}\label{eq:Herglotzvf}
    \xi_{L} = \dot{q}^i \frac{\partial}{\partial q} + h_{L}^i \frac{\partial}{\partial \dot{q}} + L \frac{\partial}{\partial z},
\end{equation}
where $h_L^i$ is the unique solution to the equation
\begin{equation}
    {\frac{\partial^2 L}{\partial \dot{q}^j \partial \dot{q}^i}}h_{L}^i + {\dot{q}^i  {\frac{\partial^2 L}{\partial \dot{q}^j \partial q^i}} 
+ {L  {\frac{\partial^2 L}{\partial \dot{q}^j \partial z}}}  \frac{\partial L}{\partial q^j} =
- \frac{\partial L}{\partial \dot{q}^j} \frac{\partial L}{\partial z}}.
\end{equation}

Let $\xi$ be a SODE on $TQ \times \RR$ with a coordinate expression
\begin{equation}
    \xi = \dot{q}^i \frac{\partial}{\partial q} + a^i \frac{\partial}{\partial \dot{q}} + b \frac{\partial}{\partial z}.
\end{equation}
Comparing with the expression for $\xi_L$, by looking at the last coefficient we immediately find out that, if $\xi = \xi_L$, then $b=L$. That is, a SODE can only be the Herglotz vector field of its last component. In addition, by comparing the rest of the coordinates, one can see that the SODE is a Herglotz  vector field if and only if $a^i = h_L^i$, with $L=b$. Indeed, the vector field $\xi$ is the Herglotz vector field of a regular Lagrangian if and only
\begin{subequations}
    \begin{gather}
        {\frac{\partial^2 b}{\partial \dot{q}^j \partial \dot{q}^i}}a^i + {\dot{q}^i  {\frac{\partial^2 b}{\partial \dot{q}^j \partial q^i}} 
        + {b  {\frac{\partial^2 b}{\partial \dot{q}^j \partial z}}} = \frac{\partial b}{\partial q^j}
        - \frac{\partial b}{\partial \dot{q}^j} \frac{\partial b}{\partial z}}\label{eq:inverse_herglotz},\\
        \det\left({\frac{\partial^2 L}{\partial \dot{q}^j \partial \dot{q}^i}}\right) \neq 0.\label{eq:inverse_regular}
    \end{gather} 
\end{subequations}
where the first equation is just $a^i = h_L^i$ with $L=b$, and the second one, from requiring that $L$ is regular.

This condition has a nice geometric interpretation, similar to \cite{Crampin1981a}, that we will now explore.

\subsection{The inverse problem and the Herglotz distribution}
Given a regular contact Lagrangian $L:TQ \times \RR \to \mathbb{R}$, we define its \emph{Herglotz distribution} by
\begin{equation}
    \mathcal{H}_L = \ker \eta_L.
\end{equation}

\begin{theorem}\label{thm:inv_problem_contact_distribution}
    Let $\xi$ be a SODE on $TQ\times \RR$. Then $\xi$ is a Herglotz vector field if and only if there exists a contact distribution $\mathcal{H} \subset T(TQ \times \RR)$ such that $\xi$ is an infinitesimal contactomorphism for $\mathcal{H}$, every vertical subspace of $TQ\times\RR \to Q \times \mathbb{R}$ is Legendrian and $\frac{\partial }{\partial z}$ is not tangent to $\mathcal{H}$.
\end{theorem}
\begin{proof}
    Note that if $\xi=\xi_L$ is a Herglotz vector field, then $\mathcal{H}=\mathcal{H}_L$ has all the desired properties.
    
    Conversely, let $\mathcal{H}$ is a contact distribution such that $\xi$ is an infinitesimal contactomorphism. Since every vertical subspace of $TQ\times\RR$ is Lagrangian and $\frac{\partial }{\partial z}$ is tangent to the distribution, we can see that $\mathcal{H} = \ker \eta$ for some $\eta$ of the form
    \begin{equation*}
        \eta = \dd z - y_i \dd q^i,
    \end{equation*}
    where $y_i$ are local functions. Let
    \begin{equation}
        \xi = \dot{q}^i \frac{\partial}{\partial q} + a^i \frac{\partial}{\partial \dot{q}} + b \frac{\partial}{\partial z}.
    \end{equation}
    
    Since $\xi$ is an infinitesimal conformal contactomorphism,
    \begin{equation*}
        \liedv{\xi} \eta = \dd b - \xi (y_i) \dd q^i - y_i \dd \dot{q}^i = g \eta,
    \end{equation*}
    contracting with every coordinate basis vector field,
    \begin{equation}
        \begin{cases}
            \frac{\partial b}{\partial q^i}  - \xi (y_i) &= - g y_i,\\
            \frac{\partial b}{\partial \dot{q}^i} -   y_i &= 0,\\
            \frac{\partial b}{\partial z} &= g.
        \end{cases}
    \end{equation}
    Combining these equations, we obtain
    \begin{equation}
        \xi\left(\frac{\partial b}{\partial \dot{q}^i}\right)  - \frac{\partial b}{\partial q^i}= \frac{\partial b}{\partial z} \frac{\partial b}{\partial \dot{q}^i},
    \end{equation}
    which is precisely Herglotz equation for the Lagrangian $L=b$~\eqref{eq:inverse_herglotz}. Moreover, the regularity condition~\eqref{eq:inverse_regular} comes from the fact that $\eta=\eta_L$ for $L=b$, and it is a contact form.
\end{proof}
\section{Equivalence in contact Hamiltonian systems}\label{sec:equiv_ham}
\label{sec:equivalent_contact}
In order to analyze the equivalence of Lagrangian systems, we will first systematically study the notion of equivalence on the Hamiltonian case.

Given a contact manifold $(M,\eta)$, the Hamiltonian vector field  $X_f$, where $f:M\to \RR$, is the unique infinitesimal contactomorphism such that $\eta(X_f) = -f$.

Let $(M,\eta, H)$ and $(N,\bar{\eta}, \bar{H})$ be Hamiltonian systems. We denote by $X_H$ the Hamiltonian vector field of $H:M\to\RR$ with respect to $\eta$ and by $\bar{X}_{\bar H}$ the Hamiltonian vector field of $\bar{H}:N \to \RR$ with respect to $\bar \eta$. We also denote by $\Reeb$ and $\bar{\Reeb}$ the Reeb vector fields of $\eta$ and $\bar{\eta}$, respectively.

\begin{definition}
    Two Hamiltonian systems $(M,\eta, H)$ and $(N=M,\bar{\eta}, \bar{H})$ are:
    \begin{itemize}
        \item \textit{Conformally equivalent} if $\bar{\eta}=f\eta$ and $\bar{H}=fH$  for some non-vanishing $f:M\to \RR$.
        \item \textit{Dynamically equivalent} if $X_H = {\bar X}_{\bar H}$.
    \end{itemize}

\end{definition}

When two systems are equivalent ``up to diffeomorphism'', we will say they are \textit{similar}, and the corresponding diffeomorphism will be called a \textit{similarity}. Namely, a diffeomorphism $F: M \to N$ is a: 
\begin{itemize}
    \item \emph{Strict similarity} if $F^* \bar{\eta} = \eta$ and $F^* \bar H = H$.
    \item \emph{Conformal similarity} if $F^* \bar{\eta} = f \eta$ and $F^* \bar{H} = f H$ for some non-vanishing $f:M\to \RR$.
    \item \emph{Dynamical similarity} if $F_* X_H = {\bar X}_{\bar H}$.
\end{itemize}

\begin{remark}\label{rem:factor}
Two systems are conformally (resp. dynamically) equivalent if, and only if, they are conformally (dynamically) similar with the identity as the similarity.


\end{remark}

A strict (resp.\ conformal) similarity is a strict (resp.\ conformal) contactomorphism. The converse is given by the following result.

\begin{proposition}\label{prop:conformal_dynamical}
    If $F:M \to N$ is a strict (resp.\ conformal) contactomorphism, then it is a strict (conformal) similarity if, and only if, it is a dynamical similarity.
\end{proposition}
\begin{proof}
    We prove it in the conformal case. The strict case follows by setting $f=1$.
    
    Since $F$ is a conformal contactomorphism, $F^* \bar{\eta} = f\eta$ for some non-vanishing $f:M\rightarrow\RR$. Pulling back by $F$ the equation
    \begin{equation}
        \bar{\eta}(\bar{X}_{\bar{H}}) = -\bar{H}\,,
    \end{equation}
    we obtain
    \begin{equation}
        F^* \bar{\eta}(\bar{X}_{\bar{H}}) = f \eta((F^{-1})_*(X_H)) =- F^* \bar H\,.
    \end{equation}

If $F$ is a conformal similarity, $F^* \bar H = fH$. Therefore, $f\eta{(F^{-1})_* \bar{X}_{\bar{H}}}= -fH$, and $(F^{-1})_* \bar{X}_{\bar{H}}=X_H$ because the Hamiltonian vector field for $H$ is unique. Conversely, if $(F^{-1})_* \bar{X}_{\bar{H}}=X_H$, then $F^* \bar H = -f\eta(X_H) = f H$.
\end{proof}
In particular, conformally equivalent systems are dynamically equivalent.
We will now discuss some properties of the Hamiltonian system are preserved by each kind of equivalence.

\begin{proposition}
    Every contact Hamiltonian system $(M,\eta,H)$ such that $H$ does not vanish, is conformally equivalent to $(M,\bar{\eta} = -\frac{\eta}{H},\bar{H} = -1)$, so that $X_H = \bar{X}_{\bar{H}} = \bar{\Reeb}$.
\end{proposition}

\begin{proposition}
A conformal similarity of contact Hamiltonian systems preserves the zero set of $H$. That is $F:M \to N$ maps the zero set of $H$ to the zero set of $\bar{H}$.
\end{proposition}

The zero set of the Hamiltonian has interesting geometric properties. For example, the Hamiltonian vector field can only be tangent to Legendrian manifolds whenever $H=0$~\cite{deLeon2020a}. This property is important in the thermodynamic formalism, where the equilibrium states are represented by a Legendrian submanifold of the zero set of $H$.

This result by \cite{Bravetti2020} also indicates that the Hamiltonian vector field has a special behavior at the points where $H$ vanishes. The Hamiltonian is a Reeb vector field in the set $\{H \neq 0\}$, and it is a reparametrization of the Liouville vector field in $H^{-1}(0)$. 



We remark, however,  that the dynamical behavior of the Hamiltonian vector fields at the zero set of $H$ is not necessarily different to the behavior outside it. Indeed, we will provide an example in which a dynamical equivalence does not preserve the zero set of the Hamiltonian.

Consider the following systems on $\mathbb{R}^3$: $\eta=\dd z-p\dd q$, $H=pq+z$ and $\bar{\eta}=\dd z+p\dd q$, $\bar{H}=z-pq$. One can check that $X_H=\bar{X}_{\bar{H}}$, thus they are dynamically equivalent systems (that is, the identity is a dynamical similarity), but the zero set of $H$ and $\bar{H}$ are different.
\section{Extended contact systems}\label{sec:extended_systems}

 As we argued in the introduction, we need to consider the inverse problem and equivalent Lagrangians ~\enquote{up to a change on $z$}. In order to do that, we will~\enquote{forget} about the projection $z:TQ \times \mathbb{R} \to \mathbb{R}$. We formalize this notion with the \emph{extended tangent bundle} (definition \ref{def:extended.tangent}). This section is devoted to present and explore the extended tangent bundle and how to define a contact Lagrangian formalism and its corresponding contact Hamiltonian formalism.

\subsection{Extended tangent bundles}

The main object in this section is the \emph{extended tangent bundle} over a manifold $Q$:

\begin{definition}\label{def:extended.tangent}
    An \emph{extended tangent bundle} $P$ of $Q$ is a line bundle $\rho:P \to TQ$. $\rho$ is called the \emph{mechanical state function}. We also denote by $\rho_0: P\to Q$ to the map $\rho_0 = \tau_Q \circ \rho$, where $\tau_Q:TQ \to Q$ is the canonical projection.
\end{definition}

We can think of the extended tangent bundle as the contact phase space. The projection $\rho$ provides the mechanical variables (the positions and velocities). The extra degree of freedom on the bundle represent the action. However, we do not prescribe how the action can be measured.

\begin{definition}
    An \emph{action function} of the extended tangent bundle $\rho:P\rightarrow TQ$ is a surjective map $\zeta: P \to \mathbb{R}$ such that $TP=\ker T\rho\oplus\ker T\zeta$.
\end{definition}

$$
\xymatrix{
TP \ar[rr]^{\tau_P} \ar[dd]_{T{\rho}} \ &  \ & P \ar[rr]^{\zeta} \ar[dd]_{{\rho}} \ar[ddrr]^{{\rho}_0}\ &  \   &\mathbb{R} 
\\
\ & \ & \ & \ 
\\
 T(TQ) \ar@/_1.5pc/[uu]_{\lambda}\ar[rr]^{\tau_{TQ}}\ & \ & TQ \ar[rr]^{\tau_Q}\ & \ &Q  
}
$$

Given a natural coordinate system $(q^i, \dot{q}^i)$, we can construct a coordinate system $(\rho^* q^i, \rho^* \dot{q}^i, \zeta)$ on $P$. From now on, we will abuse notation and omit the $\rho^*$ when using the coordinates on $P$. Along the text we will use two action functions $z$ and $\zeta$, with the respective coordinate systems $(q^i,\dot{q}^i,z)$ and $(q^i,\dot{q}^i,\zeta)$. Notice that the coordinate basis of vector fields on $P$ depend not only on the coordinates $(q^i, \dot{q}^i)$ on $TQ$, but also on the action functions. To make this clear, we will denote them as $\left(\frac{\partial}{\partial q^i}, \frac{\partial}{\partial \dot{q}^i},\frac{\partial}{\partial z}\right)$ and $\left(\left(\frac{\partial}{\partial q^i}\right)_\zeta, \left(\frac{\partial}{\partial \dot{q}^i}\right)_\zeta,\frac{\partial}{\partial \zeta}\right)$ respectively. They are related by:

\begin{equation}\label{eq:cfractions}
        \left(\frac{\partial }{\partial q^i}\right)_{\zeta} = 
        \frac{\partial}{\partial q^i} -  \frac{\frac{\partial \zeta}{\partial q^i}}{\frac{\partial \zeta}{\partial z}} \frac{\partial }{\partial z};\quad
        \left(\frac{\partial }{\partial \dot{q}^i}\right)_{\zeta} =
        \frac{\partial}{\partial \dot{q}^i} - \frac{\frac{\partial \zeta}{\partial \dot{q}^i}}{\frac{\partial \zeta}{\partial z}} \frac{\partial }{\partial z};\quad
        \frac{\partial \zeta}{\partial z}\frac{\partial}{\partial \zeta} = \frac{\partial}{\partial z}.
\end{equation}


We proceed to study the geometric structure of the extended tangent bundle given an action function $\zeta$. We have an isomorphism $(\rho, \zeta): P \to TQ \times \mathbb{R}$. However, only the projection onto the first factor is independent of the choice of the action function. The decomposition $TP=\ker\rho\oplus\ker\zeta$ induces a section $\lambda^\zeta$ of $T\rho$ as follows: given an element $y\in T(TQ)$, $\lambda^\zeta(y)\in TP$ is the unique element such that $T\rho(\lambda^\zeta(y))=y$ and $T\zeta(\lambda^\zeta(y))=0$. With this section we can lift the canonical elements of $TQ$ to $P$. In local coordinates we have that
$$
\lambda=\dd q^i\otimes \left(\frac{\partial }{\partial q^i}\right)_{\zeta} + \dd \dot{q}^i\otimes \left(\frac{\partial }{\partial \dot{q}^i}\right)_{\zeta}\,.
$$

\begin{definition}
    The \emph{extended almost tangent structure} on $P$ by the action function $\rho$ is the $(1,1)$ tensor field $S^\zeta=\lambda^\zeta\circ S \circ T\rho$.
    
    The \emph{extended Liouville vector field} on $P$ by the action function $\rho$ is $\Delta^\zeta=\lambda^\zeta\circ\Delta\circ \rho$.
\end{definition}

In local coordinates, $S^\zeta$ is given by
\begin{equation}
    S^\zeta = \dd {q}^i \otimes \left(\frac{\partial }{\partial \dot{q}^i} \right)_{\zeta} = \dd {q}^i \otimes \left(\frac{\partial }{\partial \dot{q}^i} -  \frac{\frac{\partial \zeta}{\partial q^i}}{\frac{\partial \zeta}{\partial z}} \frac{\partial }{\partial z}  \right).
\end{equation}
Note that $S^{\xi} = S^{\bar{\xi}}$ if and only if $\xi- \bar{\xi}$ does not depend on $\dot{q^i}$. Also, $\im S^\zeta = \ker{T \rho_0} \subseteq \ker{S}^\zeta$, and $\ker{S^\zeta}$ is an integrable rank $n+1$ distribution. The extended Liouville vector field is given by 
\begin{equation}
    \Delta^\zeta =  \dot{q}^i \left(\frac{\partial }{\partial \dot{q}^i} \right)_{\zeta} = \dot{q}^i \left(\frac{\partial }{\partial \dot{q}^i} -  \frac{\frac{\partial \zeta}{\partial q^i}}{\frac{\partial \zeta}{\partial z}} \frac{\partial }{\partial z}  \right)
\end{equation}
which is precisely the infinitesimal generator of the $\mathbb{R}$-action $\chi(a)v = \exp(a) v$ on the vector bundle $(\rho_0,\zeta):P \to Q \times \mathbb{R}$. This allows us to define SODEs on $TP$. 
\begin{definition}
    A vector field $\xi$ of $P$ is an \emph{extended SODE} (Second Order Differential Equation) if $S^\zeta (\xi) = \Delta^\zeta$ for some action function $\zeta$.
\end{definition}

An extended SODE $\xi$ has the local expression
\begin{equation}
    \xi = \dot{q}^i\frac{\partial}{\partial q^i} + a^i \frac{\partial}{\partial \dot{q}^i} + b \frac{\partial}{\partial z}\,,
\end{equation}

for $a^i, b\in C^\infty (P)$.

\begin{lemma}\label{lem:sode}
   $\xi$ is an extended SODE if, and only if, $S(\rho_*\xi_p) = \Delta_{\rho(p)}$ for all $p \in P$.
\end{lemma}
\begin{proof}
Since $\lambda$ is a section of $T\rho$:
$$S^\zeta (\xi) = \Delta^\zeta\Leftrightarrow\lambda(S(T\rho(\xi)))=\lambda(\Delta\circ\rho)\Leftrightarrow S(T\rho(\xi))=\Delta\circ\rho\,.$$
\end{proof}
Thus, the concept of extended SODE is independent on the choice of action function.

\subsection{Extended contact Lagrangian systems}
In this section we will define a generalization of contact Lagrangian systems, which will provide us a more general formulation of the problem of equivalent Lagrangians.

\begin{definition} An \emph{extended Lagrangian system} on an extended tangent bundle $\rho:P \to TQ$ consists of a Lagrangian function $L:P \to \mathbb{R}$ and an {action function} $\zeta:P \to \RR$.
\end{definition}

From an extended Lagrangian system, we can obtain the contact form
\begin{equation}
    \eta^{\zeta}_L = \dd \zeta - {(S^\zeta)}^* \dd L = 
    \dd \zeta -{ \left({\frac{\partial L}{\partial \dot{q}^i}}\right)}_{\zeta} \dd q^i.
\end{equation}

This form will be a contact form if and only if $L$ is $\zeta$-regular, that is, the matrix $( W_{ij}^\zeta)$ defined by
\begin{equation}
    W_{ij}^\zeta = {\left({\frac{\partial^2 L}{\partial \dot{q}^i \partial \dot{q}^j}}\right)}_{\zeta}
\end{equation}
is not degenerate. The regularity depends on the Lagrangian and on the action function. For instance, the Lagrangian $L=\frac12v^2-\gamma z$ is regular by the action function $\zeta=z$, but it is not regular by the action function $\zeta=\frac12v^2-\gamma z$.

The pair $(P,\eta^{\zeta}_L)$ is a contact manifold.
We denote its Reeb vector field by
\begin{equation}
    \mathcal{R}^\zeta =     \frac{\partial}{\partial \zeta}+{(W^\zeta)}^{ij}{\left({\frac{\partial^2 L}{\partial q^i \partial \zeta}}\right)}_{\zeta} \left(\frac{\partial}{\partial \dot{q}^j}\right)_\zeta
\end{equation}

Given the $\zeta$-Liouvile vector field
\begin{equation}
    \Delta^\zeta = \dot{q}^i {\left({\frac{\partial }{\partial \dot{q}^i}}\right)}_{\zeta},
\end{equation}
we define the $\zeta$-energy
\begin{equation}
    E^{\zeta}_L = \Delta^{\zeta}(L) - L\,.
\end{equation}
$(P,\eta_{L,\zeta},E_{L,\zeta})$ is a Hamiltonian contact system. The corresponding Hamiltonian vector field,  ($\zeta$-Herglotz vector field $\xi_{L,\zeta}$) is,
\begin{subequations}
    \begin{align}
        \contr{\xi_{L,\zeta}} \eta_{L,\zeta}  &= -E^{\zeta}_L,\\
        \liedv{\xi_{L,\zeta}} \eta_{L,\zeta}  &= -\Reeb_L(E^{\zeta}_L) \eta_{L,\zeta} = \frac{\partial L}{\partial \zeta} \eta_{L,\zeta}.
    \end{align}
\end{subequations}

We remark that if $\zeta = z$, this is just the usual Herglotz vector field $\xi_L$.

A variational principle may also be written for an extended Lagrangian system $(L,\zeta)$. Given the space $\Omega(q_0,q_1)$ of paths $\gamma:[0,1] \to Q$ such that $\gamma(0)=q_0$ and $\gamma(1)=q_1$, and given $\zeta_0 \in \mathbb{R}$, we define
\begin{equation}
    \mathcal{X}_{L,\zeta,\zeta_0}: \Omega(q_0,q_1) \to \mathcal{C}^\infty([0,1] \to P)
\end{equation}
such that for each curve $\gamma \in \Omega(q_0, q_1)$, $\mathcal{X}_{L,\zeta,\zeta_0}(\gamma)$ is the curve that satisfies $\rho \circ \mathcal{X}_{L,\zeta,\zeta_0}(\gamma) = \gamma'$, and its $\zeta$ component, which we denote by $\mathcal{Z}_{L,\zeta,\zeta_0}(\gamma) = \zeta \circ \mathcal{X}_{L,\zeta,\zeta_0}(\gamma)$ solves the initial value problem
\begin{equation}
    \begin{cases}
        \frac{\dd \mathcal{Z}^\zeta(\gamma)}{\dd t} &= L \circ \mathcal{X}_{L,\zeta,\zeta_0}(\gamma),\\
        \mathcal{Z}^\zeta(\gamma)(0) &= \zeta_0.
    \end{cases}
\end{equation}

We now define the action
\begin{equation}
    \begin{split}
        \mathcal{A}^{\zeta}: \Omega(q_0,q_1) &\to \mathbb{R},\\
        \gamma &\to \mathcal{Z}^\zeta(\gamma)(1) - \zeta_0 = \int_0^1 L(q^i(t), \dot{q}^i(t), \zeta(t)) \dd t.
    \end{split}
\end{equation}
By repeating the usual computation, but this time with the coordinates $(q^i, \dot{q}^i,\zeta)$, we obtain
\begin{theorem}\label{thm:ext_var_principle}
    $\gamma\in \Omega(q_0,q_1)$ is a critical point of $\mathcal{A}$ if and only if $(\gamma, \dot{\gamma}, \mathcal{Z}^{\zeta}(\gamma))$ are solutions to the $\zeta$-Herglotz equations:
    \begin{subequations}
        \begin{align}
            \left(\frac{\partial L}{\partial q^i}\right)_\zeta - \frac{\dd }{\dd t} \left(\frac{\partial L}{\partial \dot{q}^i} \right)_\zeta &= 
            \left(\frac{\partial L}{\partial \dot{q}^i}\right)_\zeta \frac{\partial L}{\partial \zeta} \\
            \dot{\zeta} &= L.     
        \end{align}
    \end{subequations}
\end{theorem}

\subsection{Extended contact Hamiltonian systems}
Now we define the Hamiltonian counterpart an extended Lagrangian system. 

An extended cotangent bundle is a line bundle $\tilde{\rho}:\tilde{P} \to T^*Q$.  A \emph{Hamiltonian action function} $\tilde{\zeta}: {P} \to \mathbb{R}$ is a surjective map such that $T \tilde{P} = \ker T \tilde{\rho} \oplus T \tilde{\zeta}$.

$$
\xymatrix{
T{{\tilde{P}}} \ar[rr]^{\tau_{{\tilde{P}}}} \ar[dd]_{T{\tilde{\rho}}} \ &  \ & {{\tilde{P}}} \ar[rr]^{\tilde{\zeta}} \ar[dd]_{{\tilde{\rho}}} \ar[ddrr]^{{\tilde{\rho}}_0}\ &  \   &\mathbb{R} 
\\
\ & \ & \ & \ 
\\
 T(T^* Q) \ar@/_1.5pc/[uu]_{\tilde{\lambda}}\ar[rr]^{\tau_{T^* Q}}\ & \ & T^* Q \ar[rr]^{\tau_Q}\ & \ &Q  
}
$$

\begin{equation}
    \eta_{Q,\tilde{\zeta}} = \dd \tilde{\zeta} - \tilde{\rho}\theta_Q = \dd \tilde{\zeta} - p_i \dd q^i.
\end{equation}

Now, given a Hamiltonian function $H: \tilde{P} \to \mathbb{R}$ we compute its Hamiltonian vector field $X_{H,\tilde{\zeta}}$ with respect to $\eta_{Q,\tilde{\zeta}}$, which is given in coordinates by
\begin{equation}
    X_H = \frac{\partial H}{\partial q^i} \left( \frac{\partial }{\partial q^i} \right)_{\tilde{\zeta}} - \left( \frac{\partial H}{\partial p_i} + p_i \frac{\partial H}{\partial z} \right) \left(\frac{\partial }{\partial p_i} \right)_{\tilde{\zeta}} + \left(p_i \frac{\partial H}{\partial p_i} - H \right) \frac{\partial }{\partial z}
\end{equation}

\subsubsection{The $\zeta$-Legendre transformation}
Given an extended Lagrangian system $(L,\zeta)$ on the extended tangent bundle $\rho:P \to TQ$, we can define its $\zeta$-Legendre transformation, which maps it to an extended Hamiltonian system. 

First, we need to construct the dual Hamiltonian bundle $\tilde{P}$ of $P$. Indeed, given the action function $\zeta$, we will be able to construct $\tilde{\rho}: \tilde{P} \to T^* Q$ and the action function $\zeta$.

We let $\rho_0 = \rho \circ \tau_Q : P \to Q$. Note that the bundle $(\rho_0, \zeta): P \to Q \times \mathbb{R} \to Q \times \mathbb{R}$ has a unique structure of a vector bundle such that the map $(\rho,\zeta): P \to TQ \times \mathbb{R}$ is a vector bundle isomorphism. We can then construct the linear dual bundle $(\tilde{\rho}, \tilde{\zeta}): \tilde{P} \to Q \times \mathbb{R}$. Now we define the map $\tilde{\rho}:\tilde{P} \to T^*Q$, such that, for any $\alpha_{q_0, \zeta_0} \in \tilde{P}$ and $v_{q_0} \in TQ$, we have that
\begin{equation}
    \tilde{\rho}(\alpha_{q_0, \zeta_0})(v_{q_0}) = \alpha_{q_0, \zeta_0}(v_{q_0,\zeta_0}),
\end{equation}
where $v_{q_0,\zeta_0} \in {P}$ is the unique element satisfying ${\rho}(v_{q_0,\zeta_0}) = v_{q_0}$ and ${\zeta}(v_{q_0,\zeta_0}) = \zeta_0$.
\begin{equation}
    \begin{tikzcd}
        P \arrow[r, "\zeta"] \arrow[rd, "\rho_0"] \arrow[d, "\rho"] & \mathbb{R} & \tilde{P} \arrow[l, "\tilde{\zeta}"'] \arrow[d, "\tilde\rho"] \arrow[ld, "\tilde{\rho}_0"] \\
        TQ \arrow[r, "\tau_Q"]                                      & Q          & T^*Q \arrow[l, "\pi_Q"]                                                                   
        \end{tikzcd}
\end{equation}

Now, we can define the $\zeta$-Legendre transform of $L$ as the fiber derivative $F^\zeta L: P \to \tilde{P}$ over the bundle $(\tilde{\rho}, \tilde{\zeta}): \tilde{P} \to Q \times \mathbb{R}$. That is
\begin{equation}
    F^\zeta L(v_{q,\zeta})(w_{q,\zeta}) =\frac{\dd }{\dd t}\vert_{t=0} L (v_{q,\zeta} + t  w_{q,\zeta}),
\end{equation}
where $v_{q,\zeta},w_{q,\zeta} \in (\tau,\zeta)^{-1}(q,\zeta)$. In local coordinates, 
\begin{equation}
    \tilde{\rho}(F^\zeta L (q,\dot{q}, \zeta)) = (q, {\left(\frac{\partial L}{\partial \dot{q}}\right)}_{\zeta})
\end{equation}
where we used coordinates $(q^i, \dot{q}^i, \zeta)$ on the right and the dual coordinates $(q^i, p^\zeta_i, \zeta)$ on the left. The map $F^\zeta L$ is a local diffeomorphism if and only if $L$ is $\zeta$-regular. If there exist $H$ is such that $F^\zeta L_* H = E_{L,\zeta}$, then $ (F^\zeta L)^* \eta_{Q,\zeta^*} = \eta_{L,\zeta}$. Hence, $F^\zeta L$ is a strict similarity for the contact systems $(TQ \times \mathbb{R},E_L^{\zeta}, \eta^{\zeta}_L)$ and $(T^*Q \times \mathbb{R}, H, \eta^{\zeta^*}_Q)$.
\section{Equivalent extended contact systems}\label{equiv.extended}

A smooth change in the $z$ variable corresponds to a change in the action function, which are realized by vector bundle automorphisms on $P$.

\begin{definition}
A \emph{horizontal diffeomorphism} is a vector bundle automorphism $\phi$ of $\rho:P \to TQ$.
\end{definition}

Given two action functions $z,\zeta$, there exists a unique horizontal diffeomorphism that satisfies the commutative  diagram
\begin{equation}\label{diagramahorizontal}
    \begin{tikzcd}
& \mathbb{R} &\\
P \arrow[rr, "\phi"] \arrow[ru, "z"] \arrow[rd, "\rho"'] &            & P \arrow[lu, "\zeta"'] \arrow[ld, "\rho"] \\
                                & {TQ} &
\end{tikzcd}
\end{equation}
and it is given by $\phi = (\Id_{TQ}, \zeta)$ on the trivialization provided by $(\rho,z)$.

In the case that $\zeta$ does not depend on the velocities, that is $\zeta = \tau^* \zeta_0$, where $ \zeta_0 : Q \times \mathbb{R}\to \mathbb{R}$ we say that $\phi$ is a \emph{strong horizontal diffeomorphism}. 

Subsequently, we analyze how the horizontal transformation act on the structures of the extended tangent bundle and how they can be used to study equivalent Lagrangians.

\subsection{Equivalence in extended tangent bundles}

A horizontal diffeomorphism $\phi$ acts on action function  with precomposition $z=\phi^*\zeta=\zeta\circ\phi$, as we can see in diagram \ref{diagramahorizontal}. The corresponding extended almost tangent structures and Liouville vector fields are not preserved by $\phi$, in general. More precisely, we have the following result.

\begin{lemma}
   If $\phi$ is a strong horizontal diffeomorphism, then
   \begin{equation}
    \phi_* S^{\phi^* \zeta} = S^{\zeta}, \quad 
    \phi_* \Delta^{\phi^* \zeta} = \Delta^{\zeta}.
\end{equation}
\end{lemma}

Now we will study the action of horizontal diffeomorphisms on extended SODEs. First we must see that is well-behaved. It turns out that preserving extended SODEs actually characterizes horizontal diffeomorphisms.
\begin{proposition}
    A vector bundle automorphism $\phi$ of $\rho_0: P \to Q$  preserves extended SODEs (that is $\phi_* \xi$ is a SODE whenever $\xi$ is a SODE) if and only if it is a horizontal diffeomorphism.
\end{proposition}
\begin{proof}
    Let, $\xi$ be a SODE and let $\phi(q^i,\dot{q}^i,z) = (q^i,\nu^i, \zeta)$. Then, using the characterization given by \cref{lem:sode},
    \begin{equation}
        S(T\rho(\phi_* \xi)) = \nu^j \frac{\partial }{\partial \dot{q}^i}\,.
    \end{equation}
    $\phi$ is a horizontal diffeomorphism if and only if $\nu^i = \dot{q}^i$. Clearly, this is the case if and only if $\phi_* \xi$ is a SODE.


\end{proof}

Since horizontal diffeomorphisms preserve SODEs, we can classify SODEs by these transformations.
\begin{definition}
    We say that two extended SODEs $\xi$ and $\bar{\xi}$ on $P$ are \emph{horizontally similar} if there exists a horizontal diffeomorphism $\phi$ such that $\phi_* \xi = \bar{\xi}$. If $\phi$ is a strong horizontal diffeomorphism, then we say that $\xi$ and $\bar{\xi}$ are \emph{strongly horizontally similar}.
\end{definition}

A direct computation shows that two SODEs, $\xi$ and $\bar{\xi}$, are horizontally similar if, and only if, there exists a function $\phi = (\Id_{TQ},  \zeta)$ that satisfies
\begin{subequations}\label{eq:equivalence_conditions}
    \begin{gather}
        a^i = \phi^* \bar{a}^i,\label{eq:equivalence_conditions1}\\
        \frac{\partial \zeta}{\partial q^i} \dot{q}^i + a^i \frac{\partial \zeta}{\partial \dot{q}^i} + b\frac{\partial \zeta}{\partial z} = \phi^* \bar{b}\label{eq:equivalence_conditions2}, \quad
        \frac{\partial \zeta}{\partial z} \neq 0,
    \end{gather}
\end{subequations}
where
\begin{subequations}\label{eq:xi_barxi_coords}
    \begin{align}\label{eq:SODE_coords}
        \xi &= \dot{q}^i \frac{\partial}{\partial q} + a^i \frac{\partial}{\partial \dot{q}} + b \frac{\partial}{\partial z},\\
        \bar{\xi} &= \dot{q}^i \frac{\partial}{\partial q} + \bar{a}^i \frac{\partial}{\partial \dot{q}} + \bar{b} \frac{\partial}{\partial z}.
    \end{align}
\end{subequations}


An interesting particular case is when a SODE $\xi$ in $P$ is \emph{$\rho$-projectable}. In the coordinates~\eqref{eq:xi_barxi_coords}, this means that $a^i$ does not depend on $z$. Horizontal equivalences preserve this property.

\begin{proposition}\label{prop:projectable}
    Let $\xi,\bar{\xi}$ be extended SODEs on $P$ and let $\xi$ be $\rho$-projectable. Then  $\xi,\bar{\xi}$ are horizontally equivalent if and only if $\bar{\xi}$ is also $\rho$-projectable and  ${\rho}_* \xi = {\rho}_* \bar{\xi}$.
\end{proposition}

\begin{proof}
    Assume that $\xi$ is projectable and horizontally equivalent to $\bar{\xi}$ coordinates~\eqref{eq:SODE_coords}, then, by~\eqref{eq:equivalence_conditions1}, by taking the inverse of $\phi$, we obtain ${(\phi^{1-})}^* a^i = a^i = \bar{a}^i$, hence $\bar{\xi}$ is $\tau_1$-projectable and ${\tau_1}_* \xi = {\rho}_* \bar{\xi}$.

    For the converse, we will see if $\xi$ is projectable, then it is horizontally equivalent to
    \begin{equation}
        \hat{\xi} = \dot{q}^i \frac{\partial}{\partial q} + \bar{a}^i \frac{\partial}{\partial \dot{q}}.
    \end{equation}
    By transitivity of the equivalence relation, this will imply that $\xi$ is horizontally equivalent to any other SODE with the same projection. 

    Using~\eqref{eq:equivalence_conditions}, we see that $\xi$ and $\hat{\xi}$ are equivalent if and only if there exists a solution for the following equation
    \begin{equation}
        \frac{\partial \zeta}{\partial q^i} \dot{q}^i + a^i \frac{\partial \zeta}{\partial \dot{q}^i}  = b, \quad
        \frac{\partial \zeta}{\partial z} \neq 0.
    \end{equation}
    Since this is a linear, first order PDE, there exist local solutions.  Since the equation only involves partial derivatives of $\zeta$ with respect to $q$ and $\dot{q}$ adding a function of $z$ to the solution, we can obtain a new one so that $\frac{\partial \zeta}{\partial z} \neq 0$ does not vanish.

\end{proof}

\subsection{Equivalent contact Lagrangian systems}
Notice that extended Lagrangian systems are pullbacks by horizontal diffeomorphisms of usual Lagrangian systems. That is, $\phi = (\Id_{TQ}, \zeta)$ is an exact similarity for the contact systems $(TQ \times \mathbb{R}, \eta_L, E_L)$ and $(TQ \times \mathbb{R}, \eta_{\phi^*{L},\zeta}  E{\phi^* {L}, \zeta})$. Indeed, we have
\begin{align*}
    \phi^*E_L &= \phi^*(\Delta(L))- \phi^*(L) = (\phi_*\Delta)(\phi^*L)- \phi^*(L) = \eta^{\zeta}_{\phi^*L} \\
    \phi^* \eta_L &= \phi^*\dd \zeta - \phi^*\left(\frac{\partial L}{\partial \dot{q}^i}\right) \dd q^i = \dd \zeta - \left(\frac{\partial L}{\partial \dot{q}^i}\right)_\xi \dd q^i,
\end{align*}
by~\cref{prop:conformal_dynamical}. By~\cref{rem:factor}, we know that given to Lagrangian systems $(L,z)$ and $(\bar{L},{\zeta})$ the map $\phi^{-1}$ is a (conformal/dynamical) equivalence between the systems $(TQ \times \mathbb{R},E_L, \eta_L)$ and $(TQ \times \mathbb{R},E_L^{\zeta}, \eta^{\zeta}_L)$.

As a consequence, $L$ and $\bar{L}$ are horizontally equivalent if and only if $\xi_L = \xi_{\bar{L}, \zeta}$ for some action function $\zeta$. Hence, in order to study the problem of equivalent Lagrangians, we can equally study the following problem.

\begin{problem}[Equivalent Lagrangians]
    Which extended Lagrangian systems $(L,z)$ $(\bar{L},\zeta)$ have the same dynamics.
\end{problem}

\begin{definition}
    Two extended Lagrangian systems  $(L,z)$ $(\bar{L},\zeta)$  are \emph{equivalent} if $\xi_{L} = \xi_{\bar{L}, \zeta}$. If $\mathbb{\zeta}$ does not depend on $\dot{q}^i$, we say that they are \emph{strongly} equivalent.
\end{definition}

\begin{theorem}\label{thm:lagrangian.equivalence}
    Let $(L,z)$ and $(\bar{L}, {\zeta})$ be regular extended Lagrangian systems. Both systems are equivalent if and only if    
    \begin{subequations}\label{eq:L_p_conditions}
        \begin{gather}            
            \bar{L} = \liedv{\xi_L} (\zeta) =  \dot{q}^i \frac{\partial \zeta}{\partial q^i} + h_L^i \frac{\partial \zeta}{\partial \dot{q}^i} + \frac{\partial \zeta}{\partial z} L \label{eq:L_condition} \\
            \xi_L (p^{\bar{L},\zeta}_i) - \left( \frac{\partial \bar{L}}{\partial q^i} \right)_\zeta =
            \left( \frac{\partial \bar{L}}{\partial \zeta} \right)_\zeta p^{\bar{L},\zeta}_i, \label{eq:p_condition}
        \end{gather}
    \end{subequations}
        where
        \begin{equation}
            p_i^{\bar{L}, \zeta} = {\left(\frac{\partial \bar{L}}{\partial \dot{q}^i}\right)}_{\zeta}.
        \end{equation}
\end{theorem}

\begin{proof}
    Assume that  both systems are equivalent. By definition $\xi_{L,z} = \xi_{\bar{L}, {\zeta}}$, hence $\bar{L} = \xi_{\bar{L}, {\zeta}} = \liedv{\xi_L} (\zeta)$. Also,~\eqref{eq:p_condition}, after changing  $\xi_{L,z}$ by $\xi_{\bar{L}, {\zeta}}$, are just the Herglotz equations for $(\bar{L}, {\zeta})$.
    
    Conversely, assume that conditions~\eqref{eq:L_p_conditions} hold. Thus, we  need to prove that both Herglotz vector fields are equal. We can do that by proving that $\xi_{L,z}$ is the Hamiltonian vector field of $\eta_{\bar{L},\zeta}$ with respect to the energy function $E_{\bar{L}}$. That is,
    \begin{align*}
        \eta_{\bar{L},\zeta} (\xi_{L,z}) &= - E_{\bar{L},\zeta}, \\
        \liedv{\xi_{L,z}} \eta_{\bar{L},\zeta} &=  \frac{\partial \bar{L}}{\partial \zeta} \eta_{\bar{L},\zeta}.
    \end{align*}
    Expanding the first equation, we obtain
     \begin{align*}
        \liedv{\xi_{L,\zeta}} (\zeta) - p_i^{\bar{L}, \zeta} \dot{q}^i = -(p_i^{\bar{L}, \zeta} \dot{q}^i - \bar{L}),
     \end{align*}
     hence it is equivalent to~\cref{eq:L_condition}. Assuming that the first condition holds, the second condition yields
     \begin{align*}
          \dd \bar{L} - \liedv{\xi_{L,z}}(p_i^{\bar{L}, \zeta}) \dd q^i -  p_i^{\bar{L}, \zeta} \dd \dot{q}^i = \frac{\partial \bar{L}}{\partial \zeta} (\dd \zeta  - p_i^{\bar{L}, \zeta} \dd {q}^i).
     \end{align*}
     Contracting with $(\partial / \partial \dot{q}^i)_\zeta$ and $\partial/\partial \zeta$ we obtain $0$ on both sides of the equation. If we contract with $(\partial / \partial {q}^i)_\zeta$, we obtain
     \begin{equation}
         {\left(\frac{\partial \bar{L}}{\partial q^i}\right)}_{\zeta} - \liedv{\xi_{L,z}}(p_i^{\bar{L}, \zeta}) = - \frac{\partial L}{\partial \zeta}p_i^{\bar{L}, \zeta},
     \end{equation}
     which is~\cref{eq:p_condition}.
\end{proof}

The notion of strong equivalence has a nice characterization: it coincides with that of conformal equivalence. Moreover, it is easy to find a closed form for these Lagrangians

\begin{theorem}\label{thm:lagrangian_strong_equivalece}
    $(L, z)$ and $(\bar{L}, \zeta)$ be regular extended Lagrangian systems. Then, the following are equivalent
    \begin{enumerate}
        \item $(L, z)$ and $(\bar{L}, \zeta)$ are strongly equivalent.
        \item  $\Id_{TQ \times \mathbb{R}}$ is a conformal similarity between $(TQ \times \mathbb{R}, \eta_L, E_L)$ and $(TQ \times \mathbb{R}, \eta^{\zeta}_L, E_{\bar{L}, \zeta})$.
        \item We have
        \begin{equation}\label{eq:lag_equiv}
            \frac{\partial \zeta}{\partial z} L + \dot{q}^i \frac{\partial \zeta}{\partial q^i} = \bar{L},
        \end{equation}
        and $\zeta$ is independent of $\dot{q}^i$.
        \end{enumerate}
\end{theorem}
\begin{proof}
    We will proof that $3 \implies 2 \implies 1 \implies 3$.

    Assume that~\eqref{eq:lag_equiv} holds. Then, 
    \begin{align*}
        \eta^\xi_{\bar{L}}&=\dd \zeta-\left(\frac{\partial \bar{L}}{\partial \dot{q}^i}\right)_\zeta \dd q^i=\frac{\partial \zeta}{\partial q^i}\dd q^i+\frac{\partial \zeta}{\partial z}\dd z-\left(\frac{\partial \zeta}{\partial z}\frac{\partial L}{\partial \dot{q}^i}+\frac{\partial \zeta}{\partial q^i}\right)\dd q^i=\frac{\partial \zeta}{\partial z}\eta_L,\\
        E^\xi_L &=\dot{q}^i \left(\frac{\partial  \bar{L}}{\partial \dot{q}^i}\right)_\zeta-\bar{L}=\dot{q}^i\left(\frac{\partial \zeta}{\partial z}\frac{\partial L}{\partial \dot{q}^i}+\frac{\partial \zeta}{\partial q^i}\right)-\frac{\partial \zeta}{\partial z} L - \dot{q}^j \frac{\partial \zeta}{\partial q^i}=\frac{\partial \zeta}{\partial z} E_L.
    \end{align*}
    Therefore, both systems are conformally equivalent.

    Now, assume that both systems are conformally equivalent. Thus, $\eta_L = f \eta_L$ for a non-vanishing $f$. We now take the contraction of the previous expression with $\partial / \partial \dot{q}^i$. We obtain
    \begin{equation}
     0  = f \frac{\partial \bar{\zeta}}{\partial \dot{q}^i}.
    \end{equation}
    Hence, both extended Lagrangian systems are strongly equivalent.

    Last of all, assume that $(L,z)$ and $(\bar{L},\zeta)$ are strongly equivalent. Then $\xi_L(\zeta)= \xi_{\bar{L}, \zeta}(\zeta)$, thus 
    \begin{equation}
        \frac{\partial \zeta}{\partial z} L + \dot{q}^i \frac{\partial \zeta}{\partial q^i} = \bar{L}.
    \end{equation}
\end{proof}

Hence, the set of strongly equivalent Lagrangians is parametrized by a function $\zeta_0: Q \times \mathbb{R} \to \mathbb{R}$.

\begin{remark}\label{remark:symplectic}
    This result includes the symplectic Lagrangian equivalence \cite{Marmo1985,Ranada1991}. Consider Lagrangians that do not depend on $z$ and take $\zeta=cz+\nu(q^i)$, with $c$ a non-zero constant. Then $\bar{L}(q^i,\dot{q}^i)=cL(q^i,\dot{q}^i)+\dot{q}^i\frac{\partial\nu}{\partial q^i}$. In particular, we have $\theta^{\zeta}_{\bar{L}}=\theta_{\bar{L}}=c\theta_L+\dd \nu$.
\end{remark}


\subsubsection{Variational formulation}

We will now analyze the problem from a variational perspective. 

Let $(L,z)$ and $(L,\zeta)$ be extended Lagrangian systems. We remind that, by \cref{thm:ext_var_principle}, the critical points of the action functional are the projections onto $Q$ of the integral curves of their Herglotz vector fields. Thus,  that if the two systems are equivalent, their corresponding action functionals must have the same critical points. But this is not sufficient, the curves on $Q$ have to be lifted to $P$, on the same way through the operator $\mathcal{X}$. Thus, we have
\begin{proposition}\label{thm:}
    Let $(L,z)$ and $(\bar{L}, \zeta)$ be Lagrangian systems. Both systems are equivalent if and only if
    \begin{enumerate}
        \item $\gamma:[0,1] \to Q$ is a critical point of $\mathcal{Z}_{L,z,z_0}$ if and only if it is a critical point of $\mathcal{Z}_{\bar{L},\zeta,\zeta_0}$, where, $\zeta_0 = \zeta(\gamma(0), \dot{\gamma}(0), z_0)$.
        \item For every critical point $\gamma$ of one has that $\mathcal{X}_{L,z,z_0}(\gamma) = \mathcal{X}_{\bar{L},\zeta,\zeta_0}(\gamma)$.
    \end{enumerate}
\end{proposition}

It will be useful for our purposes to have a geometric characterization of the operator $\mathcal{X}$.
\begin{proposition}\label{thm:X_operator_SODE}
    Let $L:P \to \mathbb{R}$ be a Lagrangian function and let $\xi$ be an extended SODE. Then, for every integral curve $\delta:[0,1] \to P$ of $\xi$, if we let $\gamma = \rho_0^* \delta :[0,1] \to Q$ and $\zeta(\delta(0)) = \zeta_0$, we have
    \begin{equation}\label{eq:liedv_X_operator}
        \mathcal{X}_{L,\zeta,\zeta_0} (\gamma) = \delta
    \end{equation}
    if and only if
    \begin{equation}
            \liedv{\xi} \zeta = L.
    \end{equation}
\end{proposition}
\begin{proof}
    We note that~\eqref{eq:liedv_X_operator} holds if and only if for every integral curve $\delta$ of $\xi$ satisfies
    \begin{equation}
        \frac{\dd (\zeta \circ \delta)}{\dd t} = {L} \circ \delta,
    \end{equation}
    while
    \begin{equation}
        \frac{\dd (\zeta \circ \mathcal{X}_{L,z,z_0}) (\gamma)}{\dd t} = {L} \circ \mathcal{X}_{L,z,z_0}(\gamma).
    \end{equation}
    Since $\zeta(\delta(0)) = \zeta_0$, by uniqueness of solution of the above ODE, we conclude  that $ \mathcal{X}_{L,\zeta,\zeta_0} (\delta) = \xi$.
\end{proof}


We can assume a stronger hypothesis regarding the action; $\mathcal{X}_{L,z,z_0} = \mathcal{X}_{\bar{L},\zeta,\zeta_0}$ not only for the critical points of the action, but for every curve. We then obtain the following.

\begin{theorem}
    Let $(L,z)$ and $(\bar{L}, \zeta)$ be Lagrangian systems. Both systems are strictly equivalent if and only if for every curve $\gamma:[0,1] \to Q$, we have
    \begin{equation}
        \mathcal{X}_{L,z,z_0}(\gamma) = \mathcal{X}_{\bar{L},\zeta,\zeta_0}(\gamma),
    \end{equation}
    where $\zeta_0 = \zeta(\gamma(0), \dot{\gamma}(0), z_0)$.
\end{theorem}
\begin{proof}
    Assume that both systems have the same $\mathcal{X}$ operators. By~\cref{thm:X_operator_SODE}
    we have that both  for every extended SODE $\xi$ such that $\liedv{\xi} z= {L}$, that is, of the form
    \begin{equation}
        \xi = \dot{q}^i \frac{\partial }{\partial q^i} + a^i \frac{\partial }{\partial \dot{q}^i} + L \frac{\partial }{\partial z}
    \end{equation}
    one has that 
    \begin{equation}
        \liedv{\xi}{\zeta} =   \dot{q}^i \frac{\partial \zeta}{\partial q^i} + a^i \frac{\partial \zeta}{\partial \dot{q}^i} + \frac{\partial \zeta}{\partial z} L  = \bar{L}.
    \end{equation}
    Since this must hold for any extended SODE of this form and the accelerations $a^i$ are arbitrary, then it is necessary that $\zeta$ does not depend on $\dot{q}^i$. Moreover, since 
    \begin{equation}\label{eq:strong_equiv_cond1}
         \dot{q}^i \frac{\partial \zeta}{\partial q^i} +  \frac{\partial \zeta}{\partial z} L  = \bar{L},
    \end{equation}
    by \cref{thm:lagrangian_strong_equivalece}, both systems are strongly equivalent. 
    
    Conversely, if both systems are strongly equivalent, again, by \cref{thm:lagrangian_strong_equivalece} we know that $\zeta$ does not depend on the velocities and that~\eqref{eq:strong_equiv_cond1} holds. Thus, for every SODE satisfying $\liedv{\xi} z = L$ one has that $\liedv{\xi} \zeta = \bar{L}$ and vice versa. By~\cref{thm:X_operator_SODE}, both systems have the same $\mathcal{X}$ operators.
\end{proof}

\section{Examples}
The next examples are some applications of the previous results in equivalent Lagrangians, exploring their limits and implications.
\subsection{Total time derivative}
In symplectic geometry, adding a total derivative, that is, of the form $h_i(q^i) \dot{q}^i$, where $h_i$ are the coefficients of an exact 1-form $h_i \dd q^i=\dd h(q^i)$, produces Lagrangians with the same dynamical equations. In contact geometry one has to be careful, because the contact equations are not linear on the Lagrangian. Nevertheless, one can proceed in a similar fashion by considering the transformation $\zeta=z+h$, resulting in the Lagrangian:
$$
\bar{L}(q^i,\dot{q}^i,\zeta)=L(q^i,\dot{q}^i,z)+\frac{\partial h}{\partial q^i} \dot{q}^i\,.
$$

The extended Lagrangian systems $(L,z)$ and $(\bar{L},\zeta)$ are equivalent by \cref{thm:lagrangian_strong_equivalece}. When the Lagrangian $L$ does not depend on $z$ we recover the usual result of the symplectic case, as explained in \cref{remark:symplectic}.

\subsection{Lorentz force}

The classical Lagrangian to describe the motion of a particle under the Lorentz force is
$$
L= \sum_{i=1}^3\left(\frac{m}{2}(\dot{q}^i)^2+kA^i\dot{q}^i\right)-k\phi\,.
$$
If one performs a change of gauge by a function $h(q^i)$ (which we will assume it is time independent), then the new Lagrangian is $L+\frac{\partial h}{\partial q^i}\dot{q}^i$. Since the difference is a total derivative both Lagrangians have the same dynamical equations.

In \cite{GM-2021} a contact version of the previous Lagrangian is considered:
$$
\tilde{L}= \sum_{i=1}^3\left(\frac{m}{2}(\dot{q}^i)^2+kA^i\dot{q}^i\right)-k\phi-\gamma z\,.
$$
A change of gauge on $\tilde{L}$ has the effect of adding a total derivative term. Since $\tilde{L}$ depends on $z$, \cref{thm:lagrangian_strong_equivalece} tells us that the new Lagrangian after the change of gauge is not strongly equivalent to $\tilde{L}$ (using the same action function $z$). This is observed explicitly in \cite{GM-2021}, where the equations of motion are derived and they turn out not to be gauge invariant. 

In order to find a gauge invariant Lagrangian description of the Lorentz force in the contact setting, in \cite{GM-2021} is proposed a generalized description of the gauge given by a triple $(\phi,\mathbf{A},f)$ which transform by a function $h$ as $(\phi-\frac{\partial h}{\partial t},\mathbf{A}+\nabla{h},f-h)$. Then, the following Lagrangian is considered
$$
	L = \sum_{i=1}^3\left(\frac{m}{2}(\dot{q}^i)^2+k\dot{q}^i\left(A^i+\frac{\partial f}{\partial q^i}\right)\right)-k\phi-\gamma z.
$$
This Lagrangian is invariant under a gauge transformation, therefore, there is no need to invoke equivalence results. Alternatively, one can consider the Lagrangian 
$$
	L = \sum_{i=1}^3\left(\frac{m}{2}(\dot{q}^i)^2+k\dot{q}^iA^i\right)-k\phi-\gamma (z+kf)\,,
$$
which, after a gauge change by $h$ (and renaming $z$ by $\zeta$) transforms into
$$
	\bar{L}(q^i,\dot{q}^i,\zeta) = \sum_{i=1}^3\left(\frac{m}{2}(\dot{q}^i)^2+k\dot{q}^iA^i\right)-k\phi-\gamma (\zeta-kh+kf)+k\frac{\partial h}{\partial q^i}\dot{q}^i\,.
$$
Using \cref{thm:lagrangian_strong_equivalece} one can check that $(L,z)$ and $(\bar{L},\zeta=z+kh)$ are strongly equivalent.

\subsection{Parachute equation}

The parachute equations models a falling object under the action of constant gravity with drag proportional to the square of the velocity:
$$
\ddot{y} - \gamma \dot{y}^2 + g = 0\,.
$$
 
 In \cite{GGMRR-2019b}  a contact Lagrangian for the parachute equation is presented.
 $$
 L=\frac12 \dot{y}^2-\frac{mg}{2\gamma}(e^{2\gamma y}-1)+2\gamma \dot{y} z\,.
 $$
 This Lagrangian has an exponential of the position in it. One wonders if there exists a more elegant strongly equivalent Lagrangian. Theorem \ref{thm:lagrangian_strong_equivalece} gives no-go results of this question.

We want to find a strongly equivalent Lagrangian which conserve the kinetic energy term, that is, with the structure:
$$
\bar{L}=\frac12 \dot{y}^2+a(y,\zeta)\dot{y}+b(y,\zeta)\,.
$$
From the \cref{thm:lagrangian_strong_equivalece}, $\zeta$ has to satisfy the identity 
$$
\frac12 \dot{y}^2+a(y,\zeta)\dot{y}+b(y,\zeta)=\dot{y}\frac{\partial\zeta}{\partial y}+\frac{\partial\zeta}{\partial z}\left(\frac12 \dot{y}^2-\frac{mg}{2\gamma}(e^{2\gamma y}-1)+2\gamma \dot{y} z \right)\,.
$$
This implies that $\frac{\partial\zeta}{\partial z}=1$, thus $\zeta=z+f(y)$. Then $a(y,\zeta)=\frac{\partial\zeta}{\partial y}+2\gamma z=f'-2\gamma f+2\gamma\zeta$ and $b(y,\zeta)=-\frac{mg}{2\gamma}(e^{2\gamma y}-1)$. Therefore, the possible Lagrangians are:
$$
\bar{L}=\frac12\dot{y}^2+\left(f'-2\gamma f+2\gamma\zeta\right)\dot{y}-\frac{mg}{2\gamma}(e^{2\gamma y}-1)\,.
$$

Thus, the exponential and the term proportional to $\dot{y}z$ are necessary for a contact Lagrangian of this type to describe the parachute equation.

\subsection{Non-strong equivalent Lagrangians}

General equivalence of extended Lagrangian systems is given by \cref{thm:lagrangian.equivalence}, where the regularity hypothesis and \ref{eq:p_condition} condition are important, as we will see in this example. Consider the Lagrangian ($\gamma\neq0$)
\begin{align*}
    L(q,\dot{q},z)=\frac12\dot{q}^2-\gamma z\,,
\end{align*}
whose Herglotz vector field is 
$$
\xi_L=\dot{q}\frac{\partial}{\partial q}-\gamma\dot{q}\frac{\partial}{\partial \dot{q}}+L\frac{\partial}{\partial z}\,.
$$
Given an action function $\zeta=z+\dot{q}^n$ (with $n\neq0$), we can use condition \eqref{eq:L_condition} to compute the potential equivalent Lagrangian:

\begin{align*}
    \bar{L}(q,\dot{q},\zeta)=-\gamma n \dot{q}^n+\frac12\dot{q}^2-\gamma\zeta+\gamma \dot{q}^n\,.
\end{align*}
In general, $(L,z)$ and $(\bar{L},\zeta=\dot{q}^n)$ are not equivalent. First, we need to check \eqref{eq:p_condition}, which in this case imposes $(n-1)^2=n-1$, therefore $n$ can only be $1$ or $2$. These leaves us with the Lagrangians
$$
\bar{L}_1=\frac12\dot{q}^2-\gamma\zeta\,;\quad \bar{L}_2=\left(\frac{1}{2}-\gamma\right)\dot{q}^2-\gamma\zeta\,.
$$
$\bar{L}_1$ is a $\zeta$-regular Lagrangian, but $\bar{L}_2$ is only $\zeta$-regular if $\gamma\neq\frac12$. In this case, they are equivalent to $(L,z)$ in virtue of \cref{thm:lagrangian.equivalence}. We can check this explicitly by computing the Herglotz vector field of $(\bar{L},\zeta=\dot{q}^n)$.
$$
\xi_{\bar{L}, \zeta}=\dot{q}\left(\frac{\partial}{\partial q}\right)_\zeta+\bar{a}\left(\frac{\partial}{\partial \dot{q}}\right)_\zeta+g\frac{\partial}{\partial \zeta}\,.
$$
Equation \eqref{eq:herglotz2} tell us that $g=\bar{L}$. Equation \eqref{eq:herglotz1} is
\begin{equation}\label{eq:ex_no_strong}
-\bar{a}(1-\gamma n(n-1)^2\dot{q}^{n-2})=\gamma \dot{q}(1-\gamma n(n-1)\dot{q}^{n-2})\,.
\end{equation}
Since $a=-\gamma \dot{q}$ (the component of the vector field $\xi_{L}$ corresponding to $\frac{\partial}{\partial \dot{q}}$) does not depend on $z$, from \cref{prop:projectable} we know that $\bar{a}=a$. Thus,~\eqref{eq:ex_no_strong} becomes
$$
(1-\gamma n(n-1)^2\dot{q}^{n-2})=(1-\gamma n(n-1)\dot{q}^{n-2})\,,
$$
which is only satisfied if $n=1,2$. For $n=2$ and $\gamma=\frac12$, \eqref{eq:ex_no_strong} becomes $0=0$ and any function $\bar{a}$ is a possible solution, a sign that the system is singular and, in particular, not equivalent to $(L,z)$.

\section{The inverse problem revisited}\label{sec:inv_problem2}

\begin{problem}[Inverse Herglotz problem]
    Given a SODE $\xi$, determine weather it is an \emph{extended Herglotz vector field}, that it, weather there exists an extended Lagrangian system $(L,\zeta)$ such that $\xi = \xi^{\zeta}_L$.
\end{problem}

\begin{theorem}\label{thm:herglotz_distribution_inverse_problem2}
    A SODE $\xi$ is an extended Herglotz vector field if and only if there exists a contact form $\eta$, with the local expression
    \begin{equation}\label{eq:inverse_problem_form}
        \eta = \dd \zeta(q, \dot{q}, z) + y_i(q,\dot{q}, z) \dd q^i,
    \end{equation}
    such that $\xi$ is an infinitesimal conformal contactomorphism for $\eta$ and $\frac{\partial \zeta}{\partial z} \neq 0$.
\end{theorem}
\begin{proof}
    Assume that $\xi$ is a conformal contactomorphism for $\eta$. Then $\phi = (\Id, \zeta)$ is a horizontal equivalence. Moreover, $\phi_* \xi$ is a conformal contactomorphism for
    \begin{equation}
        (\phi^{-1})^* \eta = \dd z - (\phi^{-1})^*(y_i) \dd q^i,
    \end{equation}
    hence, by Theorem~\ref{thm:inv_problem_contact_distribution}, $ \phi_* \xi = \xi_L$, thus $\xi = \xi_{\bar{L}, \zeta}$, where $\bar{L} = \phi^* L$. 

    Conversely, if $\xi = \xi^\zeta_{\bar L}$, then we take $\eta = \eta_{L,\zeta}$. Thus, $\liedv{\xi} \eta = \frac{\partial \bar{L}}{\partial \zeta} \eta$.
\end{proof}

\begin{corollary}
    A SODE $\xi$ is an extended Herglotz vector field for an action function $\zeta: Q\times \RR \to \RR$ if and only if
    \begin{equation}\label{eq:inv_problem_pde}
        \left(\frac{\partial \xi(\zeta)}{\partial q^i}  - \xi \left(\frac{\partial \xi(\zeta)}{\partial \dot{q}^i} \right)  \right) \frac{\partial \zeta}{\partial z}  =
         \left(\frac{\partial \zeta}{\partial q^i} 
         -\frac{\partial \xi(\zeta)}{\partial \dot{q}^i}\right)
         \frac{\partial \xi(\zeta)}{\partial z} 
    \end{equation}
    and $\frac{\partial \zeta}{\partial z} \neq 0$.
\end{corollary}
\begin{proof}
    By Theorem~\ref{thm:inv_problem_contact_distribution}, $\xi$ is horizontally equivalent to a Herglotz vector field if and only if $\liedv{\xi} \eta = g \eta$ for some $g:TQ \times \mathbb{R} \to \RR$, where 
    \begin{equation}
        \eta = \dd \zeta(q,z) + y_i(q,\dot{q}, z) \dd q^i,
    \end{equation}
    and  $\frac{\partial \zeta}{\partial z} \neq 0$.

    That is,
    \begin{equation}
        \liedv{\xi} \eta = \dd \xi(\zeta) - \xi (y_i) \dd q^i - y_i \dd \dot{q}^i = g \eta = 
        g\left(\frac{\partial \xi(\zeta)}{\partial z}\dd z + \left(\frac{\partial \zeta}{\partial q^i} -  y_i \right) \dd q^i \right) ,
    \end{equation}
    or, contacting with every coordinate basis vector field, we find that $\xi$ is horizontally equivalent to a Herglotz vector field if and only if  $\frac{\partial \zeta}{\partial z} \neq 0$ and
    \begin{equation}\label{eq:inv_problem_pdes}
        \begin{dcases}
            \frac{\partial \xi(\zeta)}{\partial q^i}  - \xi (y_i) &= g\left(\frac{\partial \zeta}{\partial q^i} -  y_i\right),\\
            \frac{\partial \xi(\zeta)}{\partial \dot{q}^i} -   y_i &= 0,\\
            \frac{\partial \xi(\zeta)}{\partial z} &= g \frac{\partial \xi(\zeta)}{\partial z}.
        \end{dcases}
    \end{equation}

    We now prove that~\eqref{eq:inv_problem_pdes} is equivalent to~\eqref{eq:inv_problem_pde}.

    Assume that $\zeta$ fulfills the equations~\eqref{eq:inv_problem_pdes}. Then, solving for $y_i$ and $g$ in the second and last equation and substituting on the first one, we obtain
    \begin{equation}
        \frac{\partial \xi(\zeta)}{\partial q^i}  - \xi \left(\frac{\partial \xi(\zeta)}{\partial \dot{q}^i} \right)   = 
        \frac{\frac{\partial \xi(\zeta)}{\partial z}}{\frac{\partial \zeta}{\partial z}}
        \left(\frac{\partial \zeta}{\partial q^i} -\frac{\partial \xi(\zeta)}{\partial \dot{q}^i}\right),
    \end{equation}
which, after reordering terms, is Equation~\eqref{eq:inv_problem_pde}.

Conversely, if there exists $\zeta$ that solves~\eqref{eq:inv_problem_pde} and  $\frac{\partial \zeta}{\partial z} \neq 0$, we define

\begin{equation}
    \begin{dcases}
        y_i &= \frac{\partial \xi(\zeta)}{\partial \dot{q}^i},\\
        g &= \frac{\frac{\partial \xi(\zeta)}{\partial z}}{\frac{\partial \zeta}{\partial z}},
    \end{dcases}
\end{equation}
so that~\eqref{eq:inv_problem_pdes} are satisfied.

\end{proof}

$$\xi = \dot{q}^i \frac{\partial}{\partial q} + b^i \frac{\partial}{\partial \dot{q}} + a \frac{\partial}{\partial z},$$
then, the local expression of \eqref{eq:inv_problem_pde} is:

$$
b^j\frac{\partial^2a}{\partial \dot{q}^i\partial \dot{q}^j}+\dot{q}^j\frac{\partial^2a}{\partial \dot{q}^i\partial q^j}-\frac{\partial a}{\partial q^i}=\frac{\partial \zeta}{\partial z}\left( \frac{\partial a}{\partial \dot{q}^i}\frac{\partial a}{\partial z}-a\frac{\partial^2a}{\partial \dot{q}^i\partial z}\right),
$$
which we will rewrite as $D_i=\frac{\partial \zeta}{\partial z} E_i$. These algebraic equations provide constraints on $a$ and $b^i$, specially because $\frac{\partial \zeta}{\partial z}$ is different form $0$ everywhere and it doesn't depend on velocities. Thus, for any point $p\in TQ\times\mathbb{R}$, we must have that

$$D_i(p)=0 \iff E_i(p)=0.$$

In the points and indices $i,j$ where they are different from $0$, we have that:

\begin{equation}
    \begin{dcases}
    \frac{D_i}{E_i}=\frac{D_j}{E_j}\,;
      \\
     \frac{\partial}{\partial \dot{q}^k}\frac{D_i}{E_i}=0\,,\quad \forall k\,.      
    \end{dcases}
\end{equation}

In the case that $a$ and $b^i$ doesn't depend on $z$, we trivially recover that $b^i$ must be a solution of the Euler-Lagrange equations of $a$, thus $\xi$ should be a symplectic Lagrangian vector field.

\section{Conclusions and further research}

In this paper we state the inverse problem and the equivalent Lagrangians problem ~\enquote{up to a change on $z$} in the contact setting. 

In order to do so we first introduce extended systems  and state the extended contact Lagrangian systems. This is a generalization of contact Lagrangian systems where the manifold is not the trivial decomposition $TQ\times \mathbb{R}$. This object allows us to consider smooth changes in the $z$ variables, which we called horizontal diffeomorphism. We show that, if a SODE is $\rho$-projectable (that is, the accelerations does not depend on $z$), then all its horizontally equivalent SODEs are also $\rho$-projectable.

Equivalent Lagrangians are defined by means of horizontal diffeomorphism. For the particular case where we have a strong horizontal equivalence (that is, it does not depend on velocities), we give different characterization of equivalent Lagrangians. We also provide several examples to explore the limits and applications of horizontal equivalence.

We provide a geometrical characterization of Herglotz vector fields similar to the one obtain by~\cite{Sarlet1982} for the symplectic framework. Some results for the inverse problem for extended Lagrangian systems are proved.

It is needed to continue the work on the inverse problem in contact Lagrangian mechanics. Even though the naive statement turns out to be trivial, in the extended contact Lagrangian systems we allow to choose the action function which give too much freedom. Nevertheless, (maybe with some appropriate extra conditions)it would be interesting to find conditions equivalent to Helmholtz's in the classical problem.

One possibility is to analyze weather a given SODE $\xi$ on $TQ$ can be extended with an extra action variable so that it can be derived from a Herglotz principle, that is $\xi = \rho_* \xi_{L, \zeta}$, from some projectable extended Herglotz vector field $\xi_{L, \zeta}$. That is, for example the case of the parachute equation, which cannot be obtained from Hamilton's principle but can be derived from the Herglotz principle.

Other topic which can be studied is the relation of this two problems with their symplectic counterparts through symplectization. 

Another interesting future research is in contact field theories. In \cite{GGMRR-2019,GGMRR-2020} the $k$-contact framework is proposed as a generalization of contact mechanics to field theories. It will be specially interesting to study equivalent Lagrangians, as the concept has been already used General Relativity to simplify Einstein-Hilbert Lagrangian.

\section*{Acknowledgements}
Manuel de León and Manuel Lainz acknowledge financial support from the Spanish Ministry of Science and Innovation (MICINN), under grants PID2019-106715GB-C21, ``Severo Ochoa Programme for Centres of Excellence in R\&D'' (CEX2019-000904-S) and from the Spanish National Research Council (CSIC), through the ``Ayuda extraordinaria a Centros de Excelencia Severo Ochoa'' (20205-CEX001). Manuel Lainz wishes to thank MICINN and the Institute of Mathematical Sciences (ICMAT) for the FPI-Severo Ochoa predoctoral contract PRE2018-083203.
\printbibliography

\end{document}